\PassOptionsToPackage{unicode}{hyperref}
\PassOptionsToPackage{hyphens}{url}
\documentclass[12pt]{article}

\usepackage{authblk}
\usepackage{setspace}
\usepackage{mathrsfs}
\usepackage{graphicx}
\usepackage{amsmath, amssymb, amsthm}
\usepackage{accents}

\newtheorem{theorem}{Theorem}[section]

\newtheorem{remark}{Remark}
\newtheorem{assump}{Assumption}
\newtheorem{assumpZ}{Assumption}
\newtheorem{assumpW}{Assumption}
\renewcommand{\theassumpZ}{Z\arabic{assumpZ}}
\renewcommand{\theassumpW}{W\arabic{assumpW}}

\newcommand*{\tran}{\top}

\DeclareFontFamily{U}{mathx}{}
\DeclareFontShape{U}{mathx}{m}{n}{<-> mathx10}{}
\DeclareSymbolFont{mathx}{U}{mathx}{m}{n}
\DeclareMathAccent{\widehat}{0}{mathx}{"70}
\DeclareMathAccent{\widecheck}{0}{mathx}{"71}

\usepackage{iftex}

\usepackage[T1]{fontenc}
\usepackage[utf8]{inputenc}
\usepackage{textcomp} 

\usepackage{xcolor}
\IfFileExists{xurl.sty}{\usepackage{xurl}}{} 
\IfFileExists{bookmark.sty}{\usepackage{bookmark}}{\usepackage{hyperref}}
\hypersetup{
  pdftitle={Disentangling Structural Breaks in Factor Models for Macroeconomic Data},
  hidelinks}
\usepackage{longtable}
\usepackage{tablefootnote}
\usepackage{threeparttablex}
\usepackage{booktabs}
\usepackage{multirow}
\usepackage{tabularx}
\usepackage{makecell}
\usepackage{afterpage}
\usepackage{calc} 

\usepackage{cleveref}

\crefname{assump}{Assumption}{Assumptions}
\Crefname{figure}{Figure}{Figures}
\crefname{assumpZ}{Assumption}{Assumptions}
\crefname{assumpW}{Assumption}{Assumptions}

\def\spacingset#1{\renewcommand{\baselinestretch}%
{#1}\small\normalsize} \spacingset{1}

\usepackage{bm}
\usepackage[shortlabels]{enumitem}

\usepackage{pdflscape}

\theoremstyle{plain}

\usepackage{physics}
\usepackage{natbib}

\newlist{propenum}{enumerate}{1} 
\setlist[propenum]{label=(\alph*), ref=\theproposition~(\alph*)}
\crefalias{propenumi}{proposition} 

\newlist{lemenum}{enumerate}{1} 
\setlist[lemenum]{label=(\alph*), ref=\thelemma~(\alph*)}
\crefalias{lemenumi}{lemma} 

\newlist{thmenum}{enumerate}{1} 
\setlist[thmenum]{label=(\alph*), ref=\thetheorem~(\alph*)}
\crefalias{thmenumi}{theorem} 

\newlist{assumpenum}{enumerate}{1} 
\setlist[assumpenum]{label=(\alph*), ref=\theassump~(\alph*)}
\crefalias{assumpenumi}{assump} 

\newlist{assumpZenum}{enumerate}{1} 
\setlist[assumpZenum]{label=(\alph*), ref=\theassumpZ~(\alph*)}
\crefalias{assumpZenumi}{assumpZ} 

\newlist{assumpWenum}{enumerate}{1} 
\setlist[assumpWenum]{label=(\alph*), ref=\theassumpW~(\alph*)}
\crefalias{assumpWenumi}{assumpW}

\newcommand{\convp}{\overset{p}{\to}}
\newcommand{\convd}{\overset{d}{\to}}

\newcommand{\limT}{\lim_{T \to \infty}}

\newcommand{\plimT}{\operatorname{plim}_{T \to \infty}}

\newcommand{\fracT}{\frac{1}{T}}

\newcommand{\fracN}{\frac{1}{N}}

\newcommand{\fracrootN}{\frac{1}{\sqrt{N}}}

\newcommand{\sumT}{\sum_{t = 1}^{T}}
\newcommand{\sumTs}{\sum_{s = 1}^{T}}

\newcommand{\sumN}{\sum_{i = 1}^{N}}
\newcommand{\sumNj}{\sum_{j = 1}^{N}}
\newcommand{\sumNk}{\sum_{k = 1}^{N}}

\newcommand{\floor}[1]{\left \lfloor #1 \right \rfloor }
\newcommand{\sumTfloor}{\sum_{t = 1}^{\floor{\pi T}}}
\newcommand{\sumTfloort}{\sum_{t = \floor{\pi T + 1}}^{T}}

\newcommand{\gt}{>}
\newcommand{\lt}{<}

\newcommand{\Op}[1]{O_p \left( #1 \right) }

\newcommand{\indicator}[1]{\mathbf{1} \left\lbrace #1 \right\rbrace }

\newcommand{\suppi}{\underset{\pi \in [\pi_1, \pi_2]}{\operatorname{sup}}}

\begin{document}
\newcommand{\blind}{1}

\if1\blind

{\title{\bf Disentangling Structural Breaks in Factor Models for Macroeconomic Data\thanks{
	We thank the editor, Atsushi Inoue, 
	an anonymous associate editor, three anonymous referees, 
	Matteo Barigozzi, Xu Han, Luke Hartigan, Daniele Massacci, Serena Ng, Mirco Rubin; 
	conferences participants at 
		2022 EC-squared, 
		2023 Society for Financial Econometrics (SoFiE), 
		2023 International Symposium on Econometric Theory and Applications (SETA), 
		2023 Chinese University of Hong Kong Workshop on Econometrics, 
		2025 Symposium of the Society for Nonlinear Dynamics and Econometrics (SNDE), 
		6th Dolomiti Macro Meetings; 
	and seminar participants at 
		City University of Hong Kong, 
		BI Norwegian Business School, 
		Seoul National University, CREST, 
		and Erasmus University Rotterdam. 
	We acknowledge the financial support of the Australian Research Council (DE200100693, DP210101440, DP220100321 and DP240100970).}}

\author{Bonsoo Koo\thanks{
		Department of Econometrics and Business Statistics, Monash University, Melbourne, Australia. 
		Emails: Koo: bonsoo.koo@monash.edu; 
		Wong: benjamin.wong@monash.edu; 
		Zhong: zeyu.zhong.1@unimelb.edu.au}}
\author{Benjamin Wong$^\dagger$}
\author{Ze-Yu Zhong\thanks{
Corresponding author. Email:
zeyu.zhong.1@unimelb.edu.au}}
  \affil{
  	$^\dagger$Monash University, 
  	$^\ddagger$University of Melbourne}
  \date{Dated: \today}
  \maketitle
} \fi

\if0\blind
{
  \bigskip
  \bigskip
  \bigskip
  \begin{center}
    {\LARGE\bf Disentangling Structural Breaks in Factor Models for Macroeconomic Data}
\end{center}
  \medskip
} \fi

\bigskip
\begin{abstract}
We develop a projection-based decomposition to disentangle structural breaks in the factor variance and factor loadings. Our approach yields test statistics that can be compared against standard distributions commonly used in the structural break literature. Because standard methods for estimating factor models in macroeconomics normalize the factor variance, they do not distinguish between breaks of the factor variance and factor loadings. Applying our procedure to U.S. macroeconomic data, we find that the Great Moderation is more naturally accommodated as a break in the factor variance as opposed to a break in the factor loadings, in contrast to extant procedures which do not tell the two apart and thus interpret the Great Moderation as a structural break in the factor loadings. Through our projection-based decomposition, we estimate that the Great Moderation is associated with an over 70\% reduction in the total factor variance, highlighting the relevance of disentangling breaks in the factor structure.
\end{abstract}

\noindent%
{\it Keywords:} factor space, structural instability, breaks, principal components, dynamic factor models
\vfill


\newpage
\spacingset{1.8} 

\section{Introduction}
Dynamic factor models are used in empirical macroeconomics and finance as a form of dimension reduction, where a large set of time series are summarized through a small number of latent factors \cite[e.g.][]{aastveit_what_2015,alessi_response_2019,barigozzi_measuring_2023}. There is now increasing evidence of structural instability with factor models estimated on macroeconomic time series \cite[see, e.g.][]{breitung_testing_2011,stock_chapter_2016}. The analysis of structural breaks in factor models presents unique challenges because breaks in the loadings and breaks in the factor variance cannot be easily disentangled. For simplicity, consider the following representative factor model common in applied work for $x_{it}, t = 1, \dots, T, i = 1, \dots, N$:
\begin{align}
\label{eqn:static_form}
x_{it} = \lambda_{i}^\tran f_t + e_{it},
\end{align}
where $\lambda_{i}$ is an $r \times 1$ vector of individual loadings, $f_t$ is an $r \times 1$ vector of factors, and $e_{it}$ is the idiosyncratic component. Because both the factors and loadings are unobserved and enter multiplicatively, a normalization is needed to separately identify them. A common approach is to normalize the variance of $f_t$. While such normalizations are often innocuous, they \textit{matter} for studying structural breaks in factor models; if breaks in the factor variance are ruled out through fixing the variance through the normalization, these must manifest as breaks in the factor loadings even if the loadings are stable. We note similar concerns have previously been raised \cite[see][]{stock_chapter_2016}.


Our contribution is to propose a projection-based decomposition to disentangle structural breaks in the factor variance and the loadings in dynamic factor models. At a high level, the projection-based decomposition reparameterizes structural breaks in the factor structure into a rotational and orthogonal component where each has an interpretation of a change in the factor covariance matrix and a change in the factor loadings, respectively.  More specifically, rotations can be thought of as some suitable twisting or stretching of the original factor space, and so are associated with the factor variance. The shift component, on the other hand, arises from recognizing that breaks in the loadings are orthogonal to the original factor space.
This insight leads to the reparameterization yielding two test statistics: (i) a test for a break in the factor covariance matrix; and (ii) a test for a break in the factor loadings. We show that these test statistics that can be compared against standard chi-square or supremum-type critical values.

In principle, while the decomposition can disentangle \textit{any} generic identifiable change in the factor structure into the two components, our work is most closely related to work which tests for structural breaks in factor models using the principal components estimator \cite[see][]{stock_forecasting_2009,breitung_testing_2011,chen_detecting_2014,han_tests_2015,baltagi_identification_2017}. Given breaks in either the factor variance and/or loadings \textit{must} manifest as breaks in the loadings in these tests, we view our procedure as being complementing, rather than supplanting, current procedures: one could always first test for breaks in the factor structure using one of these aforementioned procedures and subsequently use our insight to disentangle the break.

From the perspective of dynamic factor models often used in macroeconomics, these breaks imply different interpretations. While breaks in the loadings relate to changes in how variables relate to the factors, breaks in the factor variance imply breaks in the factor dynamics. These breaks in the factor dynamics can be in the form of a break in the dynamic process generating the factors and/or breaks in the variance of the underlying shocks to the factors. 
Confining breaks to just the loadings would therefore \textit{a priori} preclude breaks in the factor dynamics as a possible interpretation. 
While whether such interpretations are misleading or not is context-dependent, we note that breaks in the factor dynamics present a more natural interpretation of at least one historical episode where there is consensus of a general reduction in the variance across multiple macroeconomic time series: the Great Moderation. 
Indeed, in an empirical application with U.S. macroeconomic data, our tests confirm the Great Moderation as having a break in the factor covariance matrix, where through our projection-based decomposition, we estimate an over 70\% reduction in the total variance of the factors (i.e. the trace of the factor covariance matrix). 
This finding should be unsurprising to applied macroeconomists, but is an effective proof-of-concept underpinning our basic argument: only by disentangling breaks in factor variances and loadings can one attribute a break in the factor variance as part of the most well-known change in volatility common across multiple macroeconomic time series. 
We note by not disentangling the break, extant work associate breaks in the factor loadings with the Great Moderation \cite[e.g.][]{breitung_testing_2011,baltagi_estimating_2021,bai_likelihood_2024}. 
Our results suggest a more nuanced interpretation of breaks in factor models.

We proceed as follows. In Section 2, we first motivate how breaks in the factor variance and loadings manifest through an underlying dynamic factor model before introducing our projection-based decomposition. Section 3 presents the theory underpinning our tests followed by Monte Carlo simulations in Section 4. Section 5 presents an empirical application with 124 quarterly U.S. macroeconomic time series. Section 6 concludes.

\section{Interpreting and Reparameterizing Structural Breaks in Dynamic Factor Models}
We  first clarify how breaks in the factor variance and loadings imply distinct interpretations. Consider the following representative dynamic factor model \citep[e.g.][]{stock_chapter_2016}:
\begin{align}
\label{eqn:static_form2}
X_{t} &= \Lambda f_t + e_{t}, \\
\label{eqn:factor_dynamics}
f_t &= \sum_{j = 1}^{p} \Phi_j f_{t - j} + \eta_t, \quad \eta_t \sim \qty(0, \Sigma_\eta),
\end{align}
where \Cref{eqn:static_form2} stacks \Cref{eqn:static_form} across the cross-section such that $X_t$ and $e_t$ are $N \times 1$, and $\Lambda = [\lambda_{1}, \dots, \lambda_N]^\tran$ is $N \times r$. 
To focus the discussion, we omit the intercept and so the variables are either demeaned or mean zero, as per the typical setup for factor models. 
\Cref{eqn:factor_dynamics} describes the dynamics of the factors, where $\Phi_j$ are autoregressive coefficients for $f_t$, and $\eta_t$ are $r \times 1$ (reduced form) innovations with covariance $\Sigma_\eta$. 
\Cref{eqn:static_form2,eqn:factor_dynamics} clarify how changes in the factor variance and loadings imply different interpretations.\footnote{We consider what is known as the static form of the dynamic factor model. 
This suffices for our general point as the dynamic factor model has a static representation. The key difference in the general dynamic case is that there are potentially less (dynamic) factors, and that the static factors in $f_t$ in \Cref{eqn:static_form2} would now span the lags of the dynamic factors as opposed to being the factors themselves.} 
In particular, the (unconditional) covariance matrix of $f_t$ is both a function of the $\Phi_j$'s and $\Sigma_\eta$. 
Therefore, breaks in the factor variance necessitates a break in \Cref{eqn:factor_dynamics}, either in the $\Phi_j$'s, $\Sigma_\eta$, or both. 
Breaks in the $\lambda_{i}$'s, on the other hand, are isolated to breaks in the relationship between the variables and the factors in \Cref{eqn:static_form2}. 
As an illustrative example, consider the Great Moderation, an event marked by a reduction of volatility across many macroeconomic variables. 
Such an interpretation, from the perspective of the dynamic factor model, is more naturally accommodated as a reduction in the variance $f_t$ rather than multiple (proportional) breaks in the $\lambda_i$'s.

Because the factors are identified up to rotation, one needs to impose some form of normalization. 
A common approach, which we primarily focus on, is to use the principal components as it is known that the principal components estimator is able to consistently estimate the space spanned by the dynamic factors under very general conditions. 
Denoting the principal components estimator of the factors $\tilde{F} = \qty[\tilde{f}_1, \dots, \tilde{f}_T]^\tran$ and loadings $\tilde{\Lambda} = \qty[\tilde{\lambda}_1, \dots, \tilde{\lambda}_N]^\tran$, a common normalization imposes
\begin{align}
\label{eqn:norm1}
\fracN \tilde{\Lambda}^\tran \tilde{\Lambda} = V_{NT}, \quad 
\fracT \tilde{F}^\tran \tilde{F} = I_r,
\end{align}
where $V_{NT}$ is a diagonal matrix whose entries are the first $r$ eigenvalues of the covariance matrix of $X$. 
Therefore, the principal components estimator can be used to estimate both $\lambda_{i}$ and $f_t$, and the fitting and specification of \Cref{eqn:factor_dynamics} can occur as a separate step \citep[see][for more details]{stock_chapter_2016}. 
The normalization one applies for estimation convolutes the interpretation of breaks - structural break tests of the factor loadings often first estimate the factors on the full sample, then test for breaks in $\tilde{\lambda}_i$ \citep[e.g.][]{stock_forecasting_2009,breitung_testing_2011,chen_detecting_2014}. 
The necessary application of normalizations like \Cref{eqn:norm1} when estimating the factors thus makes it unclear whether finding a break in $\tilde{\lambda}_i$ is a break in $\lambda_i$, or a break in the variance of factors, since the latter is typically assumed to be unchanging as an identification condition as part of constructing the test statistic \cite[see][]{stock_chapter_2016,chen_detecting_2014}.

While not our focus, we briefly touch on other estimation methods, and how this similarly affects them. 
Though not the form of breaks that we consider, other estimation approaches such as parametric state-space methods that directly estimate \Cref{eqn:static_form2,eqn:factor_dynamics} but allowing for loadings and the factor dynamics to change, or time-varying factor models, also have to impose normalizations to meet identification conditions. 
These range from breaks in the loadings as not being common across variables \cite[e.g.][]{del_negro_dynamic_2008,abbate_changing_2016},  
to something similar to \Cref{eqn:norm1} \cite[e.g.][]{fu_estimation_2023}. 
Therefore, similar normalization issues exist whether one chooses to work with the principal components estimator or parametric state-space form, where the latter may also require the researcher to be explicit on how to separate changes in the factor variance and factor loadings as part of the identification conditions.

We note that if the goal is to establish whether there are breaks in the factor \emph{structure} (i.e. the loadings and/or the factor variance), current methods of normalizing the variance, and then testing for breaks in the loadings suffice, since any breaks in the factor variance are subsumed into the factor loadings. However, if one wanted to appropriately interpret breaks, especially from the perspective of the dynamic factor model implied by \Cref{eqn:static_form2,eqn:factor_dynamics}, breaks in the loadings and the factor variance need to be distinguished. In what follows, we present a reparameterizaton to disentangle these breaks.

\subsection{A Projection-based Decomposition to Disentangle Breaks}
We now introduce structural breaks to the dynamic factor model in \Cref{eqn:static_form2,eqn:factor_dynamics}. We only work with \Cref{eqn:static_form2} in what follows since the fitting of \Cref{eqn:factor_dynamics} can occur separately, and we can consistently estimate the space spanned by the static factors. 
Let $\pi$ denote the break fraction
which splits the data as
\begin{align}
\label{eqn:static_factor_rep}
X_{t} &= 
\begin{cases}
\Lambda_1 f_t + e_{t}, \quad \text{for} \quad t = 1, \dots, \floor{\pi T}, \\
\Lambda_2 f_t + e_{t}, \quad \text{for} \quad t = \floor{\pi T} + 1, \dots, T,
\end{cases}
\end{align}
where $f_t$ is an $r \times 1$ vector of factors, $\Lambda_1 = \qty(\lambda_{1, 1}, \dots, \lambda_{1, N})^\tran$ and $\Lambda_2 = \qty(\lambda_{2, 1}, \dots, \lambda_{2, N})^\tran$ are $N \times r$ pre- and post-break loadings, and $e_{t}$ are the idiosyncratic components. 
For expositional purposes, we treat the number of factors $r$ and the break fraction $\pi$ as known. In practice, these can be either consistently estimated or chosen \textit{a priori}, and we defer discussion on specific conditions necessary to Section 3.
By the formulation in \Cref{eqn:static_form}, any changes in the factor variance must be common across all series, whereas changes in the loadings are by definition idiosyncratic. This motivates us to decompose the change in the factor loadings, $\Lambda_2$, via a projection
\begin{align}
\label{eqn:projection}
\Lambda_2 = \Lambda_1 Z + W, 
\end{align}
where $Z$ is an $r \times r$ rotational change, and $W = \qty(w_1, \dots, w_N)^\tran$ is an $N \times r$ orthogonal shift component.
We focus on the case of $r$ being unchanging as per the extant literature, but note that a change from $r_1$ pre-break factors to $r_2$ post-break factors can be accommodated, and amounts to testing whether for a change in an appropriate sub-block of the factor variance, as detailed in Section A.6 of the Supplement.

The rotation $Z$ and orthogonal shift $W$ are associated with breaks in the factor covariance matrix and loadings,  respectively. Heuristically,  a break in the factor variance can be thought of some suitable twisting or stretching of the factors themselves, i.e. a mathematical rotation. Note that because breaks in $Z$ are breaks in the covariance matrix of the factors, they encompass breaks in both their variances and their correlations (i.e the diagonal and off-diagonals elements).
In contrast, a change in the loadings is idiosyncratic across series, and thus geometrically must lie outside and be orthogonal to the space spanned by the factors \citep[see][for similar interpretations]{wang_identification_2021,pelger_state-varying_2022,massacci_testing_2021}.

We emphasize that the projection-based decomposition can hold for \emph{any} generic structural change in the factor structure (e.g. thresholds, time-varying loadings, etc.). In what follows, we demonstrate how the reparameterizaton can be used to disentangle a one time structural \emph{break}, and naturally lead to test statistics which can disentangle changes in the factor variance and factor loadings. 

\subsection{Structural Break Setup}
Denoting the pre- and post-break sample lengths respectively as $T_1 = \floor{\pi T}$ and $T_2 = T - \floor{\pi T}$, \Cref{eqn:static_factor_rep} can be stacked in matrix form:
\begin{align}
\label{eqn:matrix_form}
X = 
\begin{bmatrix}
X_1 \\
X_2
\end{bmatrix} 
= \begin{bmatrix}
F_{1} \Lambda_1^\tran \\
F_{2} \Lambda_2^\tran 
\end{bmatrix} +
\begin{bmatrix}
e_{(1)} \\
e_{(2)}
\end{bmatrix}, 
\end{align}
where $F_{1} = (f_{1}, \dots, f_{T_1})^\tran$ are $T_1 \times r$ pre-break factors, $F_{2} = (f_{T_1 + 1}, \dots, f_{T})^\tran$ are $T_2 \times r$ post-break factors, $\Lambda_1$ and $\Lambda_2$ their respective loadings, $e_{(1)} = (e_{1}, \dots, e_{T_1})$ and $e_{(2)} = (e_{T_1 + 1}, \dots, e_T)$, and $X_1$, $X_2$ denote the respective partitions of $X$.  
By substituting \Cref{eqn:projection} into \Cref{eqn:matrix_form}, we can formulate an equivalent representation as follows 
\begin{align}
X &= 
\begin{bmatrix}
F_1 \Lambda_1^\tran \\
F_2 \left[\Lambda_1 Z + W\right]^\tran 
\end{bmatrix} + 
\begin{bmatrix}
e_{(1)} \\
e_{(2)}
\end{bmatrix} \nonumber \\
&= 
\label{eqn:projection_ert:1}
\begin{bmatrix}
F_1 & 0 \\
F_2 Z^\tran & F_2 
\end{bmatrix}
\begin{bmatrix}
\Lambda_1^\tran \\
W^\tran
\end{bmatrix} +
\begin{bmatrix}
e_{(1)} \\
e_{(2)}
\end{bmatrix} \\
X &= G \Xi^\tran + e.
\label{eqn:projection_ert:2}
\end{align}
\Cref{eqn:projection_ert:1} shows that any rotational changes induced by a non-identity $Z$ are absorbed into the factors, and any orthogonal shifts $W$ will augment the factor space. \Cref{eqn:projection_ert:2} highlights that if one were to ignore the break and use the principal components estimator over the whole sample, the estimator will be consistent for an observationally equivalent model with \emph{pseudo} factors $G$ and time invariant loadings $\Xi$. 

The observationally equivalent model with \textit{pseudo} factors plays a central role in multiple existing methods as they rely on an estimate $G$ in order to either test for breaks \cite[e.g.][]{han_tests_2015,chen_detecting_2014}, and/or estimate the break fraction \cite[e.g.][]{baltagi_identification_2017,baltagi_estimating_2021,duan_quasi-maximum_2022}. 
\Cref{eqn:projection_ert:2} shows that breaks in either $Z$ \emph{or} $W$ will induce a structural break in the pseudo factors $G$: assuming that $E\left(
f_t f_t^\tran 
\right) = \Sigma_{F}$, when $Z$ is non-identity the first $r$ columns of $G$ will correspond to the multivariate series $\left[ F_1^\tran, Z F_2^\tran \right]^\tran$ with covariance matrices $\Sigma_F$ and $Z \Sigma_F Z^\tran$,  respectively; when $W \neq \mathbf{0}$, the orthogonality of $W$ will induce the last $r$ columns of $G$ to be $\left[0, F_2^\tran\right]^\tran$, corresponding to a structural break where extra factors appear. Existing methods that rely on the \emph{pseudo} factors, $G$, will therefore detect a break if either $Z$ \emph{or} $W$ breaks, and so are unable to tell them apart. 
We note that extant work such as \cite{han_tests_2015}, \cite{baltagi_identification_2017}, and \cite{duan_quasi-maximum_2022} explicitly acknowledge this limitation, and therefore that their tests for breaks in the factor loadings may be capturing breaks in either the factor variance or factor loadings.

The use of a projection implied by Equation \eqref{eqn:projection} enables us to tell both types of breaks apart. Indeed, the case of a rotational break corresponds to $Z \neq I_r$, and can be naturally interpreted a change in the factor variance. This follows from \Cref{eqn:projection_ert:1} that the covariance matrix of the factors pre- and post-break are $\Sigma_F$ and $Z \Sigma_F Z^\tran$, respectively, which are in general different for non identity $Z$. Given that $Z$ captures changes in the factor variance, it follows that the remaining orthogonal shift where $W \neq \mathbf{0}$ must correspond to breaks in the factor loadings.\footnote{So far, we have considered the model without any intercept, as $X_t$ is typically demeaned. We briefly note that breaks in the intercepts of $X_t$ can similarly be classified as common across all series or idiosyncratic, and can therefore still be accommodated by our tests. For the rest of our paper, we take $X_t$ as being demeaned across the entire sample, like existing setups. If one wanted to control for breaks in the intercept, one possibility is to consider subsample demeaning as in \cite{bai_likelihood_2024}.}
Mechanically, these breaks lie outside the original factor space, which necessitates the estimation of more factors than necessary if one wanted to capture all the information while ignoring the break. It is this orthogonality of breaks in the loadings that underlies the ``factor augmentation'' effect in the pseudo factors $G$ raised by \cite{breitung_testing_2011}.

Disentangling these breaks entails testing for changes in these parameters: a test for a break in the factor variance corresponds to $\mathcal{H}_0: \Sigma_F = Z \Sigma_F Z^\tran$, and a test for a break in the factor loadings corresponds to $\mathcal{H}_0: W = \mathbf{0}$. 
We aim to determine which type of break has occurred, with test statistics that can be compared against standard chi-square or supremum-type critical values for the cases of a known or unknown (and thus estimated) break,  respectively. In the case of multiple breaks, practitioners can partition the data suitably and either separately or sequentially focus on different breaks, as we have done in our empirical application in \Cref{sec:empirical}.

\subsection{Implementing the Tests to Disentangle Structural Breaks}
\label{sec:break_tests:estimation}
We outline how to implement our tests to disentangle breaks in the factor variance and loadings. 
Our approach is a subsample one, and thus estimation depends on the sample splitting parameter $\pi$. However, to avoid notational complexity, we proceed by suppressing their notational dependence, deferring theoretical justifications to \Cref{sec:asym}.
The construction of our test statistics then follows for a given $\pi$ as:
\begin{enumerate}
\item Estimate the pre- and post-break factors $\tilde{F}_1$ and $\tilde{F}_2$ via principal components, subject to the normalizations $\frac{1}{T_1} \tilde{F}_1^\tran \tilde{F}_1 = \frac{1}{T_2} \tilde{F}_2^\tran \tilde{F}_2 = I_r$, i.e. $\tilde{F}_1$ is $\sqrt{T_1}$ times the first $r$ eigenvectors of $X_1 X_1^\tran$, and $\tilde{F}_2$ is $\sqrt{T_2}$ times the first $r$ eigenvectors of $X_2 X_2^\tran$. 
\item Conditional on the factors, estimate the pre- and post-break loadings via least squares as $\tilde{\Lambda}_1 = X_1^\tran \tilde{F}_1 \left( \tilde{F}_1^\tran \tilde{F}_1 \right)^{-1} = \frac{1}{T_1} X_1^\tran \tilde{F}_1$ and $\tilde{\Lambda}_2 = X_2^\tran \tilde{F}_2 \left( \tilde{F}_2^\tran \tilde{F}_2 \right)^{-1} = \frac{1}{T_2} X_2^\tran \tilde{F}_2$.
\item Estimate the rotational change and orthogonal shift as
\begin{align}
\label{eqn:Z_W_estimation}
\tilde{Z} &= \left( \tilde{\Lambda}_1^\tran \tilde{\Lambda}_1 \right)^{-1} \tilde{\Lambda}_1^\tran \tilde{\Lambda}_2, \\
\label{eqn:W_estimation}
\tilde{W} &= \tilde{\Lambda}_2 - \tilde{\Lambda}_1 \tilde{Z}.
\end{align}
\end{enumerate}
The estimates $\tilde{Z}$ and $\tilde{W}$ absorb the effects of the normalization bases present in $\tilde{\Lambda}_1$ and $\tilde{\Lambda}_2$, and hence cannot be tested directly. Instead, $\tilde{Z}$ and $\tilde{W}$ take on additional interpretations: post multiplying $\tilde{\Lambda}_1$ by $\tilde{Z}$ rotates its normalization basis to that of $\tilde{\Lambda}_2$ along with any rotational change $Z$; the remainder $\tilde{W}$ is the remaining idiosyncratic change.
%
We exploit this property in $\tilde{Z}$ to rotate the post-break factors onto the same basis as the pre-break factors,  defined as $\widehat{F} = \left[\tilde{F}_1^\tran, \tilde{Z} \tilde{F}_2^\tran \right]^\tran = \left[\widehat{f}_1, \dots, \widehat{f}_T\right]^\tran$ in order to purge out the effects of $W$, and thus use its (co)variance as a basis for a test for any rotational changes.
\subsection{Test Statistics}
\label{sec:break_tests:statistics}
We now construct the test statistics, which we label as the $Z$-test for breaks in the factor variance, and the $W$-test for breaks in the loadings following \Cref{eqn:projection}.
Our tests treat the break fraction as unknown and hence as a matter of notation are expressed as the supremum over $\qty[\pi_1, \pi_2]$ where $0 < \pi_1 < \pi < \pi_2 < 1$. 
We present the $Z$-test statistic for changes in the factor variance $\mathcal{H}_0: \Sigma_F = Z \Sigma_F Z^\tran$ across time as
\begin{align}
\label{eqn:Z_test_stats}
\suppi \mathscr{W}_Z\left( \pi, \widehat{F} \right) = 
\suppi A_Z \left( \pi, \widehat{F} \right)^\tran \widehat{S}_Z \left( \pi, \widehat{F} \right)^{-1} A_Z \left( {\pi}, \widehat{F} \right),
\end{align}
where 
\begin{align}
A_Z \left( {\pi}, \widehat{F} \right) = 
\operatorname{vech}\left( 
	\sqrt{T} \left( \frac{1}{\floor{{\pi} T}} \sum_{t = 1}^{\floor{{\pi} T}} \widehat{f}_t \widehat{f}_t^\tran - \frac{1}{T - \floor{{\pi} T}} \sum_{t = \floor{{\pi}T + 1}}^T \widehat{f}_t \widehat{f}_t^\tran \right)
\right)
\end{align} denotes the difference in the subsample means of the second moments (outer product) process of $\widehat{f}_t$, and $\operatorname{vech}()$ denotes the column-wise vectorization of a square matrix with the upper triangle excluded. Its long run variance is estimated as a weighted average of the variance from pre- and post-break data
\begin{align}
\widehat{S}_Z\left( {\pi}, \widehat{F} \right) = 
\frac{1}{{\pi}} \widehat{\Omega}_{Z, (1)} \left({\pi}, \widehat{F} \right) + \frac{1}{1 - {\pi}} \widehat{\Omega}_{Z, (2)}\left( {\pi}, \widehat{F} \right),
\end{align}
where $\widehat{\Omega}_{Z, (1)}\left({\pi}, \widehat{F} \right)$ and $\widehat{\Omega}_{Z, (2)} \left({\pi}, \widehat{F} \right)$ are HAC estimators constructed using the respective subsamples of $\operatorname{vech} \left(\widehat{f}_t \widehat{f}_t^\tran - I_r \right)$. 
Alternatively, a bootstrap based procedure developed for factor models such as \cite{djogbenou_bootstrap_2015} or \cite{goncalves_bootstrapping_2020} can be entertained.
The statistic $\suppi \mathscr{W}_Z\left( {\pi}, \widehat{F} \right)$ is a sup-Wald test for whether the subsample means of the second moments process of $\widehat{f}_t$ are the same pre- and post-break. 
We note our test statistic appears similar to that of work such as \cite{han_tests_2015} and is indeed identical when $W = \mathbf{0}$; however, a key difference is that we construct our statistic by allowing for $W \neq \mathbf{0}$.

Next, we present the $W$-test statistics for changes in the loadings $\mathcal{H}_0: W = \mathbf{0}$ across time. Defining $\tilde{w}_i$ as the $i$th row of $\tilde{W}$, the test for the $i$th series and its variance are
\begin{align}
\label{eqn:w_ind_test}
\suppi \mathscr{W}_{W, i} \left( \pi, \tilde{w}_i \right) 
&= 
\suppi T \tilde{w}_i^\tran \tilde{\Omega}_{W, i}({\pi}, \tilde{w}_i)^{-1} \tilde{w}_i, \\
\tilde{\Omega}_{W, i} \left( {\pi}, \tilde{w}_i \right) 
&= 
\frac{1}{{\pi}}\tilde{\Theta}_{1, i} \left( {\pi}, \tilde{w}_i \right) 
+ \frac{1}{1 - {\pi}} \tilde{\Theta}_{2, i} \left({\pi}, \tilde{w}_i \right),
\end{align}
which uses pre- and post-break HAC estimates of the asymptotic variance $\tilde{\Theta}_{1, i} \left( {\pi}, \tilde{w}_i \right)$ and $\tilde{\Theta}_{2, i} \left( {\pi}, \tilde{w}_i \right)$, respectively.
%
Next, define the joint statistic for all variables as:
\begin{align}
\label{eqn:w_joint_test}
\suppi \mathscr{W}_{W} \left( \pi, \tilde{w}_i \right) 
= \suppi 
(TN) \left( \frac{\sumN \tilde{w}_i}{N}\right)^\tran \tilde{\Omega}_{W} \left( {\pi}, \tilde{w}_i \right)^{-1} \left( \frac{\sumN \tilde{w}_i}{N}\right),
\end{align}
where the matrix $\tilde{\Omega}_{W} \left( {\pi}, \tilde{w}_i \right) = N^{-1} \sumN \tilde{\Omega}_{W, i} \left( {\pi}, \tilde{w}_i \right)$ is an estimate of the joint variance. 
Both the $Z$-test and $W$-tests are Wald tests; they hence have conventional $\chi^2$ critical values in the case of a known break fraction, or supremum-type critical values provided by \cite{andrews_tests_1993} in the case of an unknown break fraction.

We emphasize that our test statistics do not maintain any assumptions on their counterparts; that is, the test for $\Sigma_F = Z \Sigma_F Z^\tran$ holds irrespective of $W$, and conversely the test for $W = \mathbf{0}$ holds irrespective of $Z$. As both breaks in the factor variance and the loadings could occur simultaneously, this necessitates the practitioner to run both the $Z$-test and $W$-test to pin down the source(s) of the break. 
One can apply a straightforward Bonferroni-Holm correction to control the family wise error rate, which we do in our Monte Carlo study and empirical application. 
Finally, although we have adopted the supremum notation of the structural break literature, we note that in practice, our tests can be operationalized by simply plugging in some consistent estimate $\widehat{\pi}$. 
Thus, it is computationally much simpler, yet asymptotically identical, to simply use $\mathscr{W}_Z \left(\widehat{\pi}, \widehat{F}\right)$, $\mathscr{W}_{W, i}\left(\widehat{\pi}, \tilde{w}_i\right)$, and $\mathscr{W}_{W}\left(\widehat{\pi}, \tilde{W}\right)$, rather than manually calculating the supremum (more discussion is deferred to \Cref{sec:asym:operation}.)

\section{Asymptotic Theory}
\label{sec:asym}
We discuss the asymptotic theory underpinning the test statistics discussed in \Cref{sec:break_tests:statistics}. 
All limits are taken as both $N$ and $T$ tend to infinity simultaneously, and $\delta_{NT}$ is defined as $\min\left(\sqrt{T}, \sqrt{N} \right)$. 
For notation, 
$\norm{\cdot}$ denotes the Frobenius norm of a vector or matrix, 
$\convp$ and $\convd$ denote convergence in probability and distribution, respectively, 
$\Rightarrow$ denotes weak convergence of stochastic processes, 
$\floor{\cdot}$ denotes the floor operator, 
$M$ denotes generic finite constants, 
$tr(A)$ denotes the trace of a square matrix $A$, 
and $A^{-\tran}$ denotes the inverse transpose of any invertible matrix $A$.
\subsection{Estimation}
\label{sec:asym:estimation}
We first establish the properties of the estimated rotational and orthogonal shift components by making the following assumptions. Let $\iota_{1t} \equiv \indicator{t \leq \floor{\pi T}}$ and $\iota_{2t} \equiv \indicator{t \geq \floor{\pi T} + 1}$. 
\begin{assump}
\label{ass:1_factors}
$E\norm{f_t}^4 \lt \infty$, $E \left(f_t f_t^\tran \right) = \Sigma_F$ and $\fracT \sumT f_t f_t^\tran \convp \Sigma_F$ for some $\Sigma_F > 0$.
\end{assump}
\begin{assump}
\label{ass:2_loadings}
For $m = 1, 2$, 
$E \norm{\lambda_{m, i}}^4 \leq M$, $\norm{ \Lambda_m^{ \tran} \Lambda_m/N - \Sigma_{{\Lambda_m}}} \convp 0$ for some $\Sigma_{\Lambda_m} > 0$, and $\norm{\Lambda_{m}^\tran \Lambda_m / N - \Sigma_{\Lambda_m}} = O_p\left( N^{-1/2}\right) $. 
The shift break is orthogonal such that $\Lambda_1^\tran W = \mathbf{0}$.
\end{assump}
\begin{assump}
\label{ass:3_errors}
For all $N$ and $T$:
\begin{assumpenum}
\item \label{ass:3_errors:1} 
$E \qty(e_{it}) = 0, E \abs{e_{it}}^8 \leq M$.
\item 
$E \qty(e_s^\tran e_t /N) = E \qty(N^{-1} \sumN e_{is}e_{it}) = \gamma_{N}(s, t)$, $\left| \gamma_{N}(s, s) \right| \leq M $ for all $s$, and \\
$T^{-1}\sumT \sumTs \left| \gamma_{N}(s, t) \right| \leq M$.
\item \label{ass:3_errors:2} 
$E \qty(e_{it}e_{jt}) = \tau_{ij, t}$, with $\left| \tau_{ij, t} \right| \lt \tau_{ij}$ for some $\tau_{ij}$ and for all $t$. In addition, \\
$N^{-1} \sumN \sumNj \left| \tau_{ij}\right| \leq M$.
\item \label{ass:3_errors:3} 
$E \qty(e_{it}e_{js}) = \tau_{ij, ts}$, and $(NT)^{-1} \sumN \sumNj \sumT \sumTs \left| \tau_{ij, ts} \right| \leq M$.
\item \label{ass:3_errors:4} 
For every $(t, s)$, $E\left| N^{-1/2} \sumN \qty[e_{is}e_{it} - E \qty(e_{is}e_{it})]\right|^4 \leq M$.
\end{assumpenum}
\end{assump}
\begin{assump}
\label{ass:4_ind_groups}
For $m = 1, 2$, $\left\lbrace \lambda_{m, i} \right\rbrace$, $\left\lbrace f_t \right\rbrace $ and $\left\lbrace e_{it}\right\rbrace $ are  mutually independent groups.
\end{assump}
\begin{assump}
\label{ass:5_error_corr}
For all $T$ and $N$:
\begin{assumpenum}
\item \label{ass:5_error_corr:1} 
$\sumTs \abs{\gamma_{N}(s, t)} \leq M$,
\item \label{ass:5_error_corr:2}
$\sumNk \abs{\tau_{ki}} \leq M$.
\end{assumpenum}
\end{assump}
\begin{assump}
\label{ass:6_moments}
For all $N, T$ and $m = 1, 2$:
\begin{assumpenum}
\item \label{ass:6_moments:1} 
For each $t$, 
$E \norm{\frac{1}{\sqrt{NT}} \sumTs \sumNk f_s \qty[e_{ks} e_{kt} - E(e_{ks} e_{kt})] \cdot \iota_{ms} }^2 \leq M$,
\item \label{ass:6_moments:1a}
For each $i$, 
$E \norm{\frac{1}{\sqrt{NT}} \sumTs \sumNk \lambda_{m, k} \qty[e_{kt} e_{it} - E(e_{kt} e_{it})] \cdot \iota_{ms} }^2 \leq M$, 
\item \label{ass:6_moments:2} 
$E \norm{\frac{1}{\sqrt{NT}} \sumT \sumNk f_t \lambda_{m, k}^\tran e_{kt} \cdot \iota_{mt} }^2 \leq M$,
\item \label{ass:6_moments:3}
For each $t$
$E \norm{\frac{1}{\sqrt{N}} \sumN \lambda_{m, i} e_{it}}^4 \leq M$.
\end{assumpenum}
\end{assump}
\begin{assump}
\label{ass:7_eigen_distinct}
The eigenvalues of $\qty(\Sigma_{\Lambda_1} \Sigma_F)$ and $\qty(\Sigma_{\Lambda_2} \Sigma_F)$ are distinct.
\end{assump}
\begin{assump}
\label{ass:8_break_fraction}
For any constants $\pi_1$ and $\pi_2$ satisfying $0 < \pi_1 \leq \pi \leq \pi_2 < 1$, 
\begin{assumpenum}
\item \label{ass:8_break_fraction:1}
$
\suppi
\norm{\frac{1}{\sqrt{NT}} \sumTfloor \sumNk f_t \lambda_{m, k}^\tran e_{kt}
	\cdot  \iota_{mt}}^2 = \Op{1}$, and \\
$\suppi \norm{\frac{1}{\sqrt{NT}} \sumTfloort \sumNk f_t \lambda_{m, k}^\tran e_{kt} 
	\cdot \iota_{mt}}^2 = \Op{1}
$,
for $m = 1, 2$; and
\item \label{ass:8_break_fraction:2}
$\suppi \norm{\frac{\sqrt{T}}{\floor{\pi T}} \sumTfloor \left( f_t f_t^\tran - \Sigma_F \right)} = \Op{1}$, and \\
$\suppi \norm{\frac{\sqrt{T}}{T - \floor{\pi T}} \sumTfloort \left( f_t f_t^\tran - \Sigma_F \right)} = \Op{1}$.
\end{assumpenum}
\end{assump}
\begin{assump}
\label{ass:9:hac_conditions}
The Bartlett kernel of \cite{newey_simple_1987} is used where:
\begin{assumpenum}
\item \label{ass:9:hac_conditions:1}
The bandwidths $b_T, b_{\floor{\pi T}}$, and $b_{T - \floor{\pi T}}$ are $O \qty(T^{1/3})$; and
\item \label{ass:9:hac_conditions:2}
$\sqrt{T} / N \to 0$, and $b_T /\delta_{N, T} \to 0$ as $N, T \to \infty$.
\end{assumpenum}
\end{assump}
\Cref{ass:1_factors,ass:2_loadings,ass:3_errors,ass:4_ind_groups,ass:5_error_corr,ass:6_moments,ass:7_eigen_distinct} are either straight from or slight modifications of those in \cite{bai_inferential_2003}. 
\Cref{ass:1_factors} is  Assumption A in \cite{bai_inferential_2003}, except that we require the second moment of $f_t$ to be time invariant. This additional ``strict'' stationarity assumption is a common identification condition \citep[e.g.][and others]{han_tests_2015,baltagi_identification_2017} which limits the factors to exhibit no heteroscedasticity, but this is not restrictive in our case as changes in $\Sigma_F$ are captured by $Z$.
\Cref{ass:2_loadings} is Assumption B in \cite{bai_inferential_2003}, except that it specifies the convergence speed of $\Lambda_m^\tran \Lambda_m/N$ to be no slower than $N^{-1/2}$ for $m = 1, 2$.
\Cref{ass:2_loadings} additionally specifies that $\Lambda_1$ and $W$ are  strictly orthogonal, and is necessary to prevent the cross term $\Lambda_1^\tran W$ from acting as a potential source of ``rotational'' change \cite[see][who similarly need to assume orthogonality]{wang_identification_2021,pelger_state-varying_2022,massacci_testing_2021}.
\Cref{ass:3_errors} allows for weak serial and cross-sectional correlation and corresponds to Assumption C of \cite{bai_inferential_2003}. 
\Cref{ass:4_ind_groups} is standard in the factor modeling literature, and is the subsample version of Assumption D of \cite{bai_confidence_2006}. 
\Cref{ass:5_error_corr} strengthens \Cref{ass:3_errors}, and corresponds to Assumption E in \cite{bai_inferential_2003}. 
\Cref{ass:6_moments} are Assumptions F1-F2 of \cite{bai_inferential_2003}. Although we require \Cref{ass:6_moments} which are moment conditions in \cite{bai_inferential_2003}, asymptotic normality of $N^{-1/2} \sumN \lambda_{m, i} e_{it}$ are not required for estimation. 
\Cref{ass:6_moments:3} is slightly stronger than Assumption F3 of \cite{bai_inferential_2003}, which only requires the existence of the second moments. 
\Cref{ass:7_eigen_distinct} corresponds to Assumption G in \cite{bai_inferential_2003}. 
\Cref{ass:8_break_fraction} requires that there is infinite data pre- and post-break, corresponding to of Assumption 8 in \cite{han_tests_2015}. 
\Cref{ass:1_factors,ass:2_loadings,ass:3_errors,ass:4_ind_groups,ass:5_error_corr,ass:6_moments,ass:7_eigen_distinct,ass:8_break_fraction} suffice to establish the consistency of the subsample principal components estimates. Specifically, $\tilde{F}_1$ and $\tilde{F}_2$ are estimates of $F_1 H_1$ and $F_2 H_2$, where $H_1$ and $H_2$ are the pre- and post-break normalization bases with probability limits $H_{0, 1}$ and $H_{0, 2}$, respectively.
Using these results, the consistency of $\tilde{Z}$ and $\tilde{W}$ then follow, which are the building blocks for the $Z$- and $W$-tests.
\Cref{ass:9:hac_conditions} specifies conditions for the Bartlett kernel used in HAC estimation, where in order to eliminate the effects of factor estimation, we require $b_T / \delta_{N, T} \to 0$ in addition to the typical $\sqrt{T} / N \to 0$ rate.

\subsection{$Z$-test for Rotational Changes}
\label{sec:asym:z_test}
We analyze the $Z$-test for the null of no break in the factor variance $\mathcal{H}_0: \Sigma_F = Z \Sigma_F Z^\tran$ across time that holds regardless of $W$, and is therefore robust to breaks in the factor loadings. 
We define 
$\mathscr{W}_Z \left(\pi, FH_{0, 1} \right) = A_Z \left(\pi, FH_{0, 1} \right)^\tran \widehat{S}_Z \left( \pi, FH_{0, 1} \right)^{-1} A_Z \left(\pi, FH_{0, 1} \right)$ as the infeasible analog of 
$\mathscr{W}_Z\left( \pi, \widehat{F} \right)$,
and make the following assumptions.
\begin{assumpZ}
\label{assZ:1:HAC_assump}
\begin{assumpZenum} 
\item \label{assZ:1:HAC_assump:1}
$\Omega_Z = \limT \operatorname{Var} \left( \operatorname{vech}\left( \frac{1}{\sqrt{T}} \sumT H_{0, 1}^{\tran} f_t f_t^\tran H_{0, 1} - I_r \right) \right)$ is positive definite, and $\norm{\Omega_Z} < \infty$. Its estimators $\widehat{\Omega}_{Z, (m)} \left(\pi, FH_{0, 1} \right)$ for $m = 1, 2$ are consistent such that $\norm{\widehat{\Omega}_{Z, (m)} \left(\pi, FH_{0, 1} \right) - \Omega_Z} = o_p(1)$,
\item \label{assZ:1:HAC_assump:2} 
$\mathscr{W}_Z \left(\pi, FH_{0, 1} \right) \Rightarrow Q_p(\pi)$, 
$\suppi \mathscr{W}_Z \left(\pi, FH_{0, 1} \right) \Rightarrow \suppi Q_p(\pi)$, 
where \\
$Q_p(\pi) = \left[ B_p(\pi) - \pi B_p(1) \right]^\tran \left[ B_p (\pi) - \pi B_p(1) \right]/\left(\pi \qty(1 - \pi) \right)$, and $B_p(\cdot)$ is a $p = r(r + 1)/2$ vector of independent Brownian motions on $[0, 1]$ restricted to $[\pi_1, \pi_2] \subset (0, 1)$. 
\end{assumpZenum}
\end{assumpZ}
\Cref{assZ:1:HAC_assump:1} is a standard HAC assumption regarding the convergence of the infeasible estimators 
$\widehat{\Omega}_{Z, (1)}\left(\pi, FH_{0, 1} \right)$ and  $\widehat{\Omega}_{Z, (2)} \left(\pi, FH_{0, 1} \right)$ to their population counterpart $\Omega_Z$.
\Cref{assZ:1:HAC_assump:2} is the main result of Theorem 3 of \cite{andrews_tests_1993}, and is a necessary high-level assumption to establish the asymptotic distributions of the test statistics. 
It has been used in \cite{han_tests_2015}, for which \cite{chen_detecting_2014} provides more primitive assumptions which are satisfied for a large class of ARMA processes. 
\begin{theorem}
\label{thm:6:Z_Wald_consistency}
Under  \Cref{ass:1_factors,ass:2_loadings,ass:3_errors,ass:4_ind_groups,ass:5_error_corr,ass:6_moments,ass:7_eigen_distinct,ass:8_break_fraction,ass:9:hac_conditions,assZ:1:HAC_assump}, 
then 
$\suppi \mathscr{W}_Z\left( \pi, \widehat{F} \right) \convd \suppi Q_p(\pi)$.
\end{theorem}
\Cref{thm:6:Z_Wald_consistency} shows that the sup-$Z$ test statistic converges to a squared standardized tied down Bessel Process, for which simulated critical values can be found in \cite{andrews_tests_1993}. 
To ensure the $Z$-test's power under the alternative, we make the following assumptions.
\begin{assumpZ}
\label{assZ:2:Z_test_alter}
\begin{assumpZenum}
\item 
\label{assZ:2:Z_test_alter:1}
$\norm{Z} \lt \infty$, and $Z \Sigma_F Z^\tran \neq \Sigma_F$.
\item 
\label{assZ:2:Z_test_alter:2}$
\plimT \operatorname{inf} 
\left( \operatorname{vech}(C)^\tran \left[ \operatorname{max} \left( b_{\floor{\pi T}}, b_{T - \floor{\pi T}} \right) \widehat{S}_Z 
\left(\pi, F^* H_{0, 1} \right)^{-1} \right] \operatorname{vech}(C) \right) \gt 0$, 
where $F^* = \left[ F_1^\tran, Z F_2^\tran \right]^\tran$, and $C \equiv H_{0, 1}^\tran \qty(\Sigma_F - Z \Sigma_F Z^\tran) H_{0, 1}$.
\end{assumpZenum}
\end{assumpZ}
\Cref{assZ:2:Z_test_alter:1} formalizes the definition of a break in factor variance and is commonly assume to rule out the unlikely scenario where $Z = -I_r$, i.e. all the loadings switch their signs after the break  \cite[see][and others]{han_tests_2015,baltagi_identification_2017,baltagi_estimating_2021}.
\Cref{assZ:2:Z_test_alter:2} regulates the asymptotics of the variance of the statistics under the alternative. Together, these ensure that the subsample means of $\widehat{f}_t \widehat{f}_t^\tran$ converge to different limits, and the consistency of $\mathscr{W}_Z\left( \pi, \widehat{F} \right)$, as summarized in the following theorem.
\begin{theorem}
\label{thm:7:z_alter}
Under  \Cref{ass:1_factors,ass:2_loadings,ass:3_errors,ass:4_ind_groups,ass:5_error_corr,ass:6_moments,ass:7_eigen_distinct,ass:8_break_fraction,ass:9:hac_conditions,assZ:2:Z_test_alter}, 
$\suppi \mathscr{W}_Z\left( \pi, \widehat{F} \right) \to \infty$.
\end{theorem}

\subsection{$W$-test for Orthogonal Shifts}
\label{sec:asym:w_test}
Next, we analyze the $W$-test for breaks the loadings $\mathcal{H}_0: W = \mathbf{0}$ across time that holds regardless of $Z$, and is therefore robust to changes in the factor variance. 
First note that because $W$ is an $N \times r$ matrix where $N \to \infty$, traditional tests are infeasible. Alternative popular approaches such as Bonferroni test statistics and pooling individual test statistics can suffer from significant size distortions, whereas directly testing for a change in the number of factors requires much stricter assumptions on the error term, \cite[see][]{stock_implications_2005,han_tests_2015-1}. 
This motivates us to formulate an individual test statistic $\mathscr{W}_{W, i} \left(\pi, \tilde{w}_i \right)$ for each $i$, and a joint test statistic $\mathscr{W}_W \left(\pi, \tilde{w}_i \right)$ pooled across $N$ to overcome the infinite dimension problem, the latter of which is akin to testing whether $\fracrootN \sumN w_i$ is large, as opposed to ``small'' as defined by \cite{bates_consistent_2013}. 
To analyze the test statistics, we define the infeasible analogs of these statistics as 
$\mathscr{W}_{W, i} \left(\pi, H_{0, 2}^{-1} w_i \right)$ and 
$\mathscr{W}_W \left(\pi, H_{0, 2}^{-1} w_i \right)$, respectively, and make the following additional assumptions. 
\begin{assumpW}
\label{assW:1:sum_loadings}
For $m = 1, 2$:
\begin{assumpWenum}
\item \label{assW:1:sum_loadings:1}
$\frac{1}{\sqrt{N}} \sumN \lambda_{m, i} = \Op{1}$,
\item \label{assW:1:sum_loadings:2}
For each $t$, 
$E \norm{\frac{1}{\sqrt{N}} \sumN \lambda_{m, i} e_{it}^2}^2 \leq M$,
\item \label{assW:1:sum_loadings:3}
$E\norm{\frac{1}{N \sqrt{T}} \sumT \sum_{k \neq i} \sumN \lambda_{m, k} e_{kt} e_{it} \cdot \iota_{mt}}^2 \leq M$.
\end{assumpWenum}
\end{assumpW}
\begin{assumpW}
\label{assW:2:W_test_error_assump}
For all $N, T$,
\begin{assumpWenum}
\item \label{assW:2:W_test_error_assump:1}
For each $t$, $E \qty(\fracrootN \sumN e_{it})^2 \leq M$. 
\end{assumpWenum}
\end{assumpW}
\begin{assumpW}
\label{assW:3}
For all $N, T$, and $m = 1, 2$:
\begin{assumpWenum}
\item \label{assW:3:3}
For each $i$, $E \norm{\frac{1}{\sqrt{T}} \sumT f_t e_{it} \cdot \iota_{mt}}^4 \leq M$,
\item \label{assW:3:4} 
$E \norm{
\frac{1}{\sqrt{N T}} \sumT \sumN f_t e_{it} \cdot \iota_{mt}
}^2 \leq M$.
\end{assumpWenum}
\end{assumpW}
\begin{assumpW}
\label{assW:4:clt}
\begin{assumpWenum}
\item \label{assW:4:clt:1} 
$\mathscr{W}_{W, i} \left( \pi, H_{0, 2}^{-1} w_i \right) \Rightarrow Q_{r}(\pi)$, $\suppi \mathscr{W}_{W, i} \left( \pi, H_{0, 2}^{-1} w_i \right) \Rightarrow \suppi Q_{r}(\pi)$,
\item \label{assW:4:clt:2}
$\mathscr{W}_{W} \left( \pi, H_{0, 2}^{-1} w_i \right) \Rightarrow Q_{r}(\pi)$, $\suppi \mathscr{W}_{W} \left( \pi, H_{0, 2}^{-1} w_i \right) \Rightarrow \suppi Q_{r}(\pi)$. 
\end{assumpWenum}
\end{assumpW}
\Cref{assW:1:sum_loadings} is required to bound the sum of the loadings by $\Op{\sqrt{N}}$, and is a slightly modified version of the Assumption 7 in \cite{han_tests_2015-1}. This will hold if the loadings are centered around zero, and their sum diverges at the rate of $\sqrt{N}$ by the central limit theorem (CLT). Although somewhat stricter than a conventional factor model setup, it seems to hold for empirically used datasets.
\Cref{assW:2:W_test_error_assump} is the pooled version of \Cref{ass:3_errors}. 
\Cref{assW:3:3,assW:3:4} strengthen \Cref{ass:6_moments:1},
corresponding to Assumptions 6 e) and 6 f) of \cite{han_tests_2015-1}, and are not restrictive because they involve zero mean random variables. 
\Cref{assW:4:clt:1,assW:4:clt:2} are high-level assumptions to establish the asymptotic distributions of the infeasible test statistics, where the former is a functional CLT equivalent of Assumption F of \cite{bai_inferential_2003}, and the latter is a cross-sectionally averaged version.\footnote{These can be satisfied by a wide range of processes detailed in \cite{andrews_tests_1993} and \cite{bai_estimating_1998}, and follow by recognizing that $H_{0, 2}^{-\tran} w_i$ corresponds to the structural break coefficients on $\qty[0, F_2^\tran]^\tran$, the additional columns in the pseudo-factors which arise only under the alternative. } 
This somewhat restricts the cross-sectional correlation in $e_{it}$, but can be optionally handled by the CS-HAC estimator of \cite{bai_confidence_2006}. 
\begin{theorem}
\label{thm:W_test_distribution}
Under \Cref{ass:1_factors,ass:2_loadings,ass:3_errors,ass:4_ind_groups,ass:5_error_corr,ass:6_moments,ass:7_eigen_distinct,ass:8_break_fraction,ass:9:hac_conditions} and \Cref{assW:2:W_test_error_assump,assW:3}, 
\begin{thmenum}
\item \label{thm:W_test_distribution:1}
If \Cref{assW:4:clt:1} also holds,
$\suppi \mathscr{W}_{W, i} (\pi, \tilde{w}_i) \convd \suppi Q_r(\pi)$ for each $i$, and
\item \label{thm:W_test_distribution:2} 
If \Cref{assW:1:sum_loadings,assW:4:clt:2} also hold,
$\suppi \mathscr{W}_{W} (\pi, \tilde{w}_i) \convd \suppi Q_r(\pi)$.
\end{thmenum}
\end{theorem}
\Cref{thm:W_test_distribution} establishes the asymptotic distribution of the $W$-test statistics.
To establish the consistency of the joint $W$-test, we make the following assumption.
\begin{assumpW}
\label{assW:5:W_test_alter_assump}
There exist constants $0 < \alpha \leq 0.5$ and $C > 0$ such that as $N, T \to \infty$,
$ 
\operatorname{Pr} \left( \norm{\frac{T^{\alpha/2}}{\sqrt{N}} \sumN w_i} > C\right) \to 1
$.
\end{assumpW}
\Cref{assW:5:W_test_alter_assump} is similar to Assumption 9 of \cite{han_tests_2015-1}, and requires $\norm{\frac{T^{\alpha/2}}{\sqrt{N}} \sumN w_i}$ to be bounded away from zero asymptotically. 
Note that if $N^{-1} \sumN w_i \convp 0$ under the alternative, then $N^{-1/2} \sumN w_i$ converges in distribution to some Gaussian random variable by the CLT, and hence $\norm{N^{-1/2 + \epsilon} \sumN w_i}$ is diverging as $N \to \infty$ for any fixed $\epsilon > 0$. 
In order for $\norm{\frac{T^{\alpha/2}}{\sqrt{N}} \sumN w_i}$ to be bounded away from zero, any $\alpha \in (0, 0.5]$ such that $T^{\alpha/2} \geq N^{\epsilon}$ is required, which is not difficult as $N \propto T$ in most applications. \Cref{assW:5:W_test_alter_assump} therefore ensures the consistency of the joint test statistic under the alternative hypothesis, even if $N^{-1} \sumN w_i \convp 0$, as summarized in the following Theorem.
\begin{theorem}
\label{thm:11:W_test_alter_cons}
If 
the alternative $\mathcal{H}_1: W \neq \mathbf{0}$ holds, then as $N, T \to \infty$:
\begin{thmenum}
\item \label{thm:11:W_test_alter_cons:1} 
Under \Cref{ass:1_factors,ass:2_loadings,ass:3_errors,ass:4_ind_groups,ass:5_error_corr,ass:6_moments,ass:7_eigen_distinct,ass:8_break_fraction,ass:9:hac_conditions} 
and 
\Cref{assW:2:W_test_error_assump,assW:3}, 
then $\suppi \mathscr{W}_{W, i} \qty(\pi, \tilde{w}_i) \to \infty$, 
if $w_i \neq 0$,
\item \label{thm:11:W_test_alter_cons:2}
Under \Cref{ass:1_factors,ass:2_loadings,ass:3_errors,ass:4_ind_groups,ass:5_error_corr,ass:6_moments,ass:7_eigen_distinct,ass:8_break_fraction,ass:9:hac_conditions} 
and \Cref{assW:2:W_test_error_assump,assW:3,assW:1:sum_loadings,assW:5:W_test_alter_assump}, 
then 
$\suppi \mathscr{W}_{W} \qty(\pi, \tilde{w}_i) \to \infty$, 
if $\frac{\sqrt{N}}{T^{1 - \alpha/2}} \to 0$.
\end{thmenum}
\end{theorem}

\subsection{Operationalizing with An Estimated Break Fraction}
\label{sec:asym:operation}
In practice, instead of manually computing many tests and taking their maximum, these tests can be readily operationalized with some consistent estimate $\widehat{\pi}$ satisfying $\abs{\widehat{\pi} - \pi} = \Op{T^{-1}}$, as this results in an asymptotically identical test statistic \citep[see for example Section 4.4 of][]{bai_estimating_1998}.
In our setting, we are able to show that this condition on $\widehat{\pi}$ is a necessary and sufficient condition to ensure that the quantities $\widehat{F} \left(\widehat{\pi}\right)$ and $\tilde{W} \left(\widehat{\pi}\right)$ that depend on $\widehat{\pi}$ are still consistent at the typical $\Op{\delta_{NT}^{-2}}$ rate, and therefore suffices to establish the asymptotic equivalence of 
$\mathscr{W}_Z \left( \widehat{\pi}, \widehat{F} \left( \widehat{\pi} \right) \right)$, 
$\mathscr{W}_{w, i} \left( \widehat{\pi}, \tilde{w}_i \left(\widehat{\pi} \right) \right)$, and 
$\mathscr{W}_{W} \left( \widehat{\pi}, \tilde{w}_i \left(\widehat{\pi} \right) \right)$, 
to their respective counterparts
$\suppi \mathscr{W}_{Z} \left( \pi, \widehat{F} \right)$, 
$\suppi \mathscr{W}_{w_i} \left( \pi, \tilde{w}_i \right)$, and 
$\suppi \mathscr{W}_{W} \left( \pi, \tilde{w}_i \right)$ (see Section A.5 of the Supplement).
This condition is not restrictive, as there exist many estimators $\widehat{\pi}$ which attain this rate \citep[e.g.][]{baltagi_identification_2017,bai_likelihood_2024}.
\begin{remark}
\label{rem:r_estimation}
The number of factors $r$ in either subsample can be consistently estimated conditional on a consistent estimate of $\pi$ \citep[e.g.][]{bai_determining_2002,onatski_determining_2010,ahn_eigenvalue_2013,baltagi_identification_2017}. 
If the pre- and post-break estimates $\tilde{r}_1$ and $\tilde{r}_2$ differ, this can be accommodated by allowing $\tilde{Z}$ to be $\tilde{r}_1 \times \tilde{r}_2$ (see Section A.6 of the Supplement).
\end{remark}
\section{Monte Carlo Study}
\label{sec:monte_carlo}
\subsection{Simulation Specification}
We focus on $r = 3$ factors, and draw $\lambda_{1i}$ and $\lambda_{2i}^*$ from $N \qty(\mathbf{0}_3, I_3)$ to form $\Lambda_1$ and $\Lambda_2^*$. 
Then, $W$ is generated as $1.5 \times \qty(\Lambda_2^* - (\Lambda_1^\tran \Lambda_1)^{-1} \Lambda_1^\tran \Lambda_2^*)$. 
To simulate a break under the null $\mathcal{H}_0: W = \mathbf{0}$, we set the $\floor{\sqrt{N}}$ to $N$th rows of $W$ to be 0, and leave them unchanged to simulate a break under the alternative, similar to \cite{massacci_testing_2021}.
For the rotation, under the null $\mathcal{H}_0: Z \Sigma_F Z^\tran = I$ we set $Z$ to $I_3$; otherwise we set $Z$ as a lower triangular matrix with $[2.5, 1.5, 0.5]$ on its diagonal and lower triangular entries drawn from $N(0, 1)$ as in \cite{duan_quasi-maximum_2022} to simulate a break under the alternative. 
Each $x_{it}$ is then generated as
\begin{align}
\label{eqn:overall_sim}
x_{it} = 
\begin{cases}
\lambda_{1, i}^\tran f_t + \sqrt{\theta} e_{it}, &\quad t = 1, \dots, \floor{\pi T}, \\
\qty(Z \lambda_{1, i} + w_{i})^\tran f_t + \sqrt{\theta} e_{it}, &\quad t = \floor{\pi T} + 1, \dots, T,
\end{cases}
\end{align}
for $i = 1, \dots, N$ and $t = 1, \dots, T$. The parameter $\theta$ is controlled to set the signal-to-noise ratio to be $\sim 50\%$. The factors and errors are generated as:
\begin{align}
f_{k, t} &= \rho f_{k, t - 1} + \mu_{it}, \quad \mu_{it} \sim i.i.d. \; N \qty(0, 1 - \rho^2), \\
e_{it} &= \alpha e_{i, t - 1} + v_{it},
\end{align}
where $\rho \in \left\lbrace 0, 0.7 \right\rbrace$ controls autocorrelation in the factors, and $\mu_{it}$, $v_{it}$ are mutually independent with $v_{t} = (v_{1, t}, \dots, v_{N, t})^\tran$ being i.i.d. $N \qty(0, \Omega)$ for $t = 1, \dots, T$. 
The scalar $\alpha \in \left\lbrace 0, 0.3 \right\rbrace $ allows mild serial correlation, and $\Omega_{ij} = \beta^{\abs{i - j}}$ with $\beta \in \left\lbrace 0, 0.3 \right\rbrace $ allows mild cross-sectional correlation \cite[see][]{bates_consistent_2013,baltagi_identification_2017}.
The break fraction is set to $0.5$ and estimated using the method of \cite{baltagi_identification_2017} for each of the $1000$ replications. 
Due to the sup-$Z$ and sup-$W$ tests requiring two separate fits of principal components, we set the trimming parameter to be 0.3 in order to ensure sufficient data in each subsample. 
Disentanglement requires running both tests, which could lead to a higher family wise error rate. 
Thus, we report results obtained with
the unadjusted $p$-values, and the Bonferroni-Holm adjusted $p$-values, which controls the overall Type 1 error rate by 
comparing the lower $p$-value in each pair of tests against half the significance level,
for each replication. 
\subsection{Simulation Results}
\begin{table}
	\scriptsize
	\caption{\label{tab:size_small}Size of sup-$Z$ and sup-$W$ Tests, $N = 200, r = 3$}
	\centering
	\begin{threeparttable}
		\begin{tabular}[t]{ccccccccc}
			\toprule
			\multicolumn{4}{c}{ } & \multicolumn{2}{c}{sup-Z Test} & \multicolumn{2}{c}{sup-W Test} & \multicolumn{1}{c}{sup-W Individual} \\
			\cmidrule(l{3pt}r{3pt}){5-6} \cmidrule(l{3pt}r{3pt}){7-8} \cmidrule(l{3pt}r{3pt}){9-9}
			$T$ & $\rho$ & $\alpha$ & $\beta$ & Unadj. & Adj. & Unadj. & Adj. &  \\
			\midrule
			&  & 0.0 & 0.0 & 0.174 & 0.134 & 0.012 & 0.004 & 0.030\\
			
			\multirow{-2}{*}[0.5\dimexpr\aboverulesep+\belowrulesep+\cmidrulewidth]{\centering\arraybackslash 200} &  & 0.3 & 0.3 & 0.186 & 0.138 & 0.014 & 0.010 & 0.028\\
			
			&  & 0.0 & 0.0 & 0.103 & 0.059 & 0.064 & 0.050 & 0.050\\
			
			\multirow{-2}{*}[0.5\dimexpr\aboverulesep+\belowrulesep+\cmidrulewidth]{\centering\arraybackslash 500} & \multirow{-4}{*}[1.5\dimexpr\aboverulesep+\belowrulesep+\cmidrulewidth]{\centering\arraybackslash 0.0} & 0.3 & 0.3 & 0.093 & 0.062 & 0.052 & 0.042 & 0.049\\
			\cmidrule{1-9}
			&  & 0.0 & 0.0 & 0.150 & 0.098 & 0.024 & 0.014 & 0.037\\
			
			\multirow{-2}{*}[0.5\dimexpr\aboverulesep+\belowrulesep+\cmidrulewidth]{\centering\arraybackslash 200} &  & 0.3 & 0.3 & 0.142 & 0.100 & 0.024 & 0.014 & 0.033\\
			
			&  & 0.0 & 0.0 & 0.078 & 0.046 & 0.118 & 0.082 & 0.055\\
			
			\multirow{-2}{*}[0.5\dimexpr\aboverulesep+\belowrulesep+\cmidrulewidth]{\centering\arraybackslash 500} & \multirow{-4}{*}[1.5\dimexpr\aboverulesep+\belowrulesep+\cmidrulewidth]{\centering\arraybackslash 0.7} & 0.3 & 0.3 & 0.072 & 0.046 & 0.074 & 0.054 & 0.049\\
			\bottomrule
		\end{tabular}
		\begin{tablenotes}
			\item \textit{Note: } 
			\item Entries report the rejection frequencies for the sup-$Z$-test for break in factor variance, and sup-$W$-test for break in factor loadings. Nominal size is 5\%. The parameters $\alpha$ and $\beta$ denote the degree of serial and cross-sectional correlation in the idiosyncratic component, respectively, whereas $\rho$ denotes the degree of autocorrelation in the factors.
		\end{tablenotes}
	\end{threeparttable}
\end{table}
We present the size analysis in \Cref{tab:size_small}. Across all specifications where $T > N$, the sup-$Z$ test has a nominal size close to 5\% regardless of the serial and cross-sectional correlation in the errors. The sup-$W$ test appears to suffer from some undersizing; however, this problem is isolated to small $T$. For larger $T$, implementation of the Bonferroni-Holm procedure to adjust the $p$-values also seems to correct any oversizing, so we advocate for its use. 

\Cref{tab:z_w_power} presents the power of the sup-$Z$ and sup-$W$ tests across all types of breaks; both have good power and are rejecting correctly only on their respective break types. This contrasts with the tests of \cite{han_tests_2015} and \cite{baltagi_estimating_2021}, which reject across all break types, and thus cannot discern which type of break has occurred. 

\begin{table}
	\scriptsize
	\caption{\label{tab:z_w_power}Power of sup-$Z$ and sup-$W$ Tests, $r = 3$, $N = 200$, $\alpha = \beta = 0.3$}
	\centering
	\begin{threeparttable}
		\begin{tabular}[t]{cccccccccc}
			\toprule
			\multicolumn{3}{c}{ } & \multicolumn{2}{c}{sup-Z Test} & \multicolumn{3}{c}{sup-W Test} & \multicolumn{1}{c}{ } \\
			\cmidrule(l{3pt}r{3pt}){4-5} \cmidrule(l{3pt}r{3pt}){6-8}
			Break Type & $T$ & $\rho$ & Unadj. & Adj. & Unadj. & Adj. & Individual & HI & BKW\\
			\midrule
			&  & 0.0 & 0.070 & 0.070 & 0.994 & 0.994 & 0.492 & 0.880 & 0.712\\
			
			& \multirow{-2}{*}[0.5\dimexpr\aboverulesep+\belowrulesep+\cmidrulewidth]{\centering\arraybackslash 200} & 0.7 & 0.036 & 0.036 & 0.992 & 0.992 & 0.568 & 0.984 & 0.964\\
			
			&  & 0.0 & 0.042 & 0.042 & 1.000 & 1.000 & 0.798 & 1.000 & 1.000\\
			
			\multirow{-4}{*}[1.5\dimexpr\aboverulesep+\belowrulesep+\cmidrulewidth]{\centering\arraybackslash $W \neq 0$} & \multirow{-2}{*}[0.5\dimexpr\aboverulesep+\belowrulesep+\cmidrulewidth]{\centering\arraybackslash 500} & 0.7 & 0.039 & 0.039 & 1.000 & 1.000 & 0.800 & 1.000 & 1.000\\
			\cmidrule{1-10}
			&  & 0.0 & 1.000 & 1.000 & 0.006 & 0.006 & 0.041 & 1.000 & 1.000\\
			
			& \multirow{-2}{*}[0.5\dimexpr\aboverulesep+\belowrulesep+\cmidrulewidth]{\centering\arraybackslash 200} & 0.7 & 1.000 & 0.996 & 0.014 & 0.014 & 0.063 & 1.000 & 0.996\\
			
			&  & 0.0 & 1.000 & 1.000 & 0.038 & 0.038 & 0.057 & 1.000 & 1.000\\
			
			\multirow{-4}{*}[1.5\dimexpr\aboverulesep+\belowrulesep+\cmidrulewidth]{\centering\arraybackslash $Z \neq I$} & \multirow{-2}{*}[0.5\dimexpr\aboverulesep+\belowrulesep+\cmidrulewidth]{\centering\arraybackslash 500} & 0.7 & 1.000 & 1.000 & 0.060 & 0.060 & 0.066 & 1.000 & 1.000\\
			\cmidrule{1-10}
			&  & 0.0 & 1.000 & 1.000 & 1.000 & 1.000 & 0.594 & 1.000 & 1.000\\
			
			& \multirow{-2}{*}[0.5\dimexpr\aboverulesep+\belowrulesep+\cmidrulewidth]{\centering\arraybackslash 200} & 0.7 & 1.000 & 1.000 & 0.998 & 0.998 & 0.619 & 1.000 & 1.000\\
			
			&  & 0.0 & 1.000 & 1.000 & 1.000 & 1.000 & 0.805 & 1.000 & 1.000\\
			
			\multirow{-4}{*}[1.5\dimexpr\aboverulesep+\belowrulesep+\cmidrulewidth]{\centering\arraybackslash $W \neq 0$ and $Z \neq I$} & \multirow{-2}{*}[0.5\dimexpr\aboverulesep+\belowrulesep+\cmidrulewidth]{\centering\arraybackslash 500} & 0.7 & 1.000 & 1.000 & 1.000 & 1.000 & 0.809 & 1.000 & 1.000\\
			\bottomrule
		\end{tabular}
		\begin{tablenotes}
			\item \textit{Note: } 
			\item Entries denote the rejection rates across different simulated break types; $W \neq 0$ denotes a break in the factor loadings, $Z \neq I$ a break in the factor variance, and $W \neq 0$ and $Z \neq I$ denoting a break in both. HI and BKW denote the test statistics of Han and Inoue (2015) and Baltagi et al (2021), respectively. See \Cref{tab:size_small} for explanation of $\alpha, \beta$ and $\rho$.
		\end{tablenotes}
	\end{threeparttable}
\end{table}

\section{Empirical Application}
\label{sec:empirical}

We apply our methodology to FRED-QD, a standard U.S. quarterly macroeconomic dataset \citep[see][]{mccracken_fred-qd_2020}, developed to mimic various iterations of the \cite{stock_forecasting_2009} dataset; both are popular in the factor modeling literature to test and date breaks \cite[e.g.][]{chen_detecting_2014,baltagi_estimating_2021}.

We consider the sample 1959Q3-2019Q4 from the FRED-QD dataset, and adopt the suggested data cleaning and transformations, removing top level aggregates to yield a panel of 124 series. We leave additional details of the data in Table B.1 of the Supplement. 
The start date aligns with extant work, and we end before the COVID-19 pandemic, due to the current lack of consensus on how to deal with the COVID-19 pandemic in factor models \cite[e.g., see][]{ng_modeling_2021,stock_comovement_2021,han_global_2024}. We consider two breaks in 1984Q1 and 2008Q3, as estimated by the procedure of \cite{baltagi_estimating_2021} using 3-to-6 factors and a trimming parameter of 0.1, which we henceforth refer to as the Great Moderation and Great Recession break, respectively.\footnote{These  are the breaks which the \cite{baltagi_estimating_2021} procedure tends to estimate, and are only statistically significant for 3-to-6 factors.} 
These breaks align with known events and existing evidence: the Great Moderation break is consistent with work that use a sample of 1960 to the mid 2000s and find a break in the early 1980s \cite[e.g.][]{stock_forecasting_2009,breitung_testing_2011,chen_detecting_2014,baltagi_estimating_2021}, and the Great Recession break has some evidence of breaks in loadings as noted by \cite{stock_disentangling_2012}, \cite{barigozzi_simultaneous_2018}, and \cite{cheng_shrinkage_2016}.

The number of factors in each subsample differs across estimators and we generally find somewhere between 2-to-4 factors in each sub-regime. 
We do not find any evidence that the number of factors changed pre-and-post Great Moderation, though the exact number differs between estimators. 
The \cite{bai_determining_2002}, \cite{onatski_determining_2010}, and \cite{ahn_eigenvalue_2013} procedures find 4, 3, and 1 factors, respectively. Pre-and-post Great Recession, the evidence is mixed whether the number of factors changed. The \cite{onatski_determining_2010} procedure finds 3 factors pre-and-post Great Recession, whereas \cite{bai_determining_2002} and \cite{ahn_eigenvalue_2013} find that the number of factors change from 4-to-6 and 1-to-2 respectively. 
We thus proceed as follows. Given our baseline procedure, like the extant literature, treats the number of factors as being the same pre- and post-break, we consider 1-to-4 factors since this aligns with the number of factors we find in each regime when the number of factors do not change. Nonetheless, as the \cite{bai_determining_2002} and \cite{ahn_eigenvalue_2013} procedures suggest the possibility that the number of factors change pre-and-post Great Recession, we also consider a changing number of factors, as per the modification of our baseline procedure described in Section A.6 of the Supplement.


\subsection{Joint test results}
\Cref{tab:joint} reports the $p$-values for the sup-$Z$ and joint sup-$W$ tests, which test for breaks in the factor variance and factor loadings, respectively, with the Great Moderation and the Great Recession break dates estimated by the \cite{baltagi_estimating_2021} procedure. 
For the Great Moderation break, we can reject the sup-$Z$ test, but not the joint $W$-test, consistent with the idea that the Great Moderation is a break in the factor variance, but not the factor loading.
For the Great Recession break, the results are somewhat mixed. 
From 1-to-4 factors, we cannot reject the sup-$Z$ test for 1 or 2 factors, but can reject for 3 or 4 factors. 
We can always reject the $W$-test, and can therefore conclude that there is evidence that the loading change. 
Given the \cite{onatski_determining_2010} procedure finds 3 factors, and is the only one that finds the number of factors do not change pre-and-post Great Recession, this would suggest both the factor variance and factor loadings change. 
If we allow for the number of factors to change, either from 1-to-2 \citep{ahn_eigenvalue_2013} or 3-to-6 \citep{bai_determining_2002}, the results are similarly mixed whether there is a break in the factor variance.

\begin{table}
\scriptsize
\caption{\label{tab:joint}Joint Test Results}
\centering
\begin{threeparttable}
\begin{tabularx}{\textwidth}{c}
\centering
\hspace{0.20\textwidth}\begin{tabular}[t]{ccccc}
\toprule
\multicolumn{1}{c}{ } & \multicolumn{2}{c}{sup-$Z$ Test $p$-values} & \multicolumn{2}{c}{sup-$W$ Test $p$-values} \\
\cmidrule(l{3pt}r{3pt}){2-3} \cmidrule(l{3pt}r{3pt}){4-5}
$\tilde{r}$ & Unadjusted & Adjusted & Unadjusted & Adjusted\\
\midrule
\addlinespace[0.3em]
\multicolumn{5}{l}{\textbf{Great Moderation (1984 Q1), 1959 Q3 - 2008 Q3 Sample}}\\
\hspace{1em}1 & 0.030 & 0.060 & 0.985 & 0.985\\
\hspace{1em}2 & 0.021 & 0.043 & 0.407 & 0.407\\
\hspace{1em}3 & 0.000 & 0.000 & 0.485 & 0.485\\
\hspace{1em}4 & 0.000 & 0.000 & 0.202 & 0.202\\
\addlinespace[0.3em]
\multicolumn{5}{l}{\textbf{Great Recession (2008 Q3), 1984 Q2 - 2019 Q4 Sample}}\\
\hspace{1em}1 & 1.000 & 1.000 & 0.031 & 0.062\\
\hspace{1em}2 & 0.109 & 0.109 & 0.000 & 0.001\\
\hspace{1em}3 & 0.013 & 0.013 & 0.000 & 0.000\\
\hspace{1em}4 & 0.000 & 0.000 & 0.000 & 0.000\\
\hspace{1em}3 -> 6 & 0.004 & 0.004 & 0.000 & 0.000\\
\hspace{1em}1 -> 2 & 1.000 & 1.000 & 0.001 & 0.003\\
\bottomrule
\end{tabular}
\end{tabularx}
\begin{tablenotes}
\item \textit{Note: } 
\item Rejection of the sup-$Z$ test corresponds to a break in the factor covariance matrix, and rejection of the sup-$W$ test corresponds to a break in the loadings across the entire cross-section. For the Great Moderation sample, Bai and Ng (2002) and Onatski (2010) estimate 3 factors pre- and post-break, whereas Ahn and Horenstein (2011) estimates 1 factor pre- and post-break. For the Global Recession sample, Bai and Ng (2002) estimates 3 and 6 factors pre- and post-break, Onatski (2010) estimates 3 factors pre- and post-break, and Ahn and Horenstein (2011) estimates 1 and 2 factors pre- and post-break.
\end{tablenotes}
\end{threeparttable}
\end{table}
%

A key tenet of the argument of the importance of distinguishing breaks in the factor variance and loadings is because the routine normalization applied in factor models rules out the possible interpretation that the factor variance changed; breaks in the factor variance are necessarily subsumed into breaks in the factor loadings. Given the clear rejection of the sup-$Z$ test for the Great Moderation break, our results suggest that rotational breaks, which may stem from changes in the factor variance, are important when modeling U.S. macroeconomic data using factors.


Beyond testing, our decomposition also permits an estimate of the total variance pre-and-post break, $tr \left( \Sigma_F \right)$ and $tr \left(Z \Sigma_F Z^\tran \right)$, respectively. 
%
In order to quantify this change, we present an estimate of the ratio, $tr \left( H_{0, 1}^\tran Z \Sigma_F Z^\tran H_{0, 1} \right) / tr \left( H_{0, 1}^\tran \Sigma_F H_{0, 1}\right)$, in \Cref{tab:sigma_f}, along with a 95\% confidence interval.
We estimate that the Great Moderation was associated with an over 70\% reduction in the total variance of the factors, across the specification of 1 to 4 factors. 
For the Great Recession break, this ratio is close to 1, with 1 being within the confidence interval for all except specifications except 4 factors, which could be due to the effects of extra factor appearing as evidenced by the \cite{bai_determining_2002} criterion. 
Thus, coupled with the mixed evidence of the sup-$Z$ test for the Great Recession break, we conclude that breaks in the factor variance were less important for understanding the Great Recession break relative to the Great Moderation break.\footnote{Note that rejection of the sup-$Z$ test corresponds to a break in the factor covariance matrix, and thus may occur even when the total factor variance (trace of the variance matrix) remains stable. This can occur if the individual factors' variances (diagonal elements) break but their sum remains the same, if the correlation between factors (off-diagonals) change, or some combination thereof. Since we can only estimate the space spanned by the factors, the sup-$Z$ test cannot tell the two cases apart. All we can claim is given the trace appears similar pre- and post-break with the Great Recession, breaks in factor variances are less important, relative to the Great Moderation.} 
We conclude that the first two factors had a sizeable role in the total reduction of the variance when looking at specific factors.\footnote{The Great Moderation was associated with the variance of the first two factors falling by about 70\% and the third and fourth factor falling by over 85\%. Given the estimate of the trace is quite similar as we move from 1-4 factors, this fall is thus largely accounted for by the first two factors.} 
While we caveat and caution that one should not over-interpret factors obtained via principal components, the first factor obtained via principal components is usually interpreted as a real activity factor and the second factor representing nominal quantities like prices \cite[see, e.g.][]{stock_forecasting_2002}. 
If one attaches such an interpretation to these factors, we can conclude that the reduction in the total factor variance associated with the Great Moderation are likely associated with the broader macro-economy, as opposed to isolated idiosyncratic groups of variables.

\begin{table}
\scriptsize
\caption{\label{tab:sigma_f}$tr(H_{0, 1}^\tran Z \Sigma_F Z^\tran H_{0, 1}) / tr(H_{0, 1}^\tran \Sigma_F H_{0, 1})$  Ratio Estimates}
\centering
\begin{threeparttable}
\begin{tabularx}{0.8\textwidth}{c}
\hspace{0.1\textwidth} \begin{tabular}[t]{ccc}
\toprule
$r$ & $tr(\tilde{Z} \tilde{Z}^\tran) / tr(I_r)$ & 95\% Bootstrap CI\\
\midrule
\addlinespace[0.3em]
\multicolumn{3}{l}{\textbf{Great Moderation (1984 Q1), 1959 Q3 - 2008 Q3 Sample}}\\
\hspace{1em}1 & 0.297 & {}[0.229, 0.345]\\
\hspace{1em}2 & 0.244 & {}[0.18, 0.305]\\
\hspace{1em}3 & 0.227 & {}[0.215, 0.279]\\
\hspace{1em}4 & 0.216 & {}[0.196, 0.282]\\
\addlinespace[0.3em]
\multicolumn{3}{l}{\textbf{Great Recession (2008 Q3), 1984 Q2 - 2019 Q4 Sample}}\\
\hspace{1em}1 & 1.040 & {}[0.507, 1.428]\\
\hspace{1em}2 & 1.180 & {}[0.881, 1.497]\\
\hspace{1em}3 & 1.588 & {}[0.913, 2.016]\\
\hspace{1em}4 & 1.255 & {}[1.206, 1.556]\\
\hspace{1em}3 -> 6 & 1.641 & {}[0.989, 2.103]\\
\hspace{1em}1 -> 2 & 1.695 & {}[0.947, 2.295]\\
\bottomrule
\end{tabular}
\end{tabularx}
\begin{tablenotes}
\item \textit{Note: } 
\item Confidence intervals obtained via a block bootstrap procedure with a block size of $m = 8$.
\end{tablenotes}
\end{threeparttable}
\end{table}
Overall, our results would suggest that the Great Moderation is more likely a break in the factor variance than a break in the factor loadings. 
This nuances a number of existing results \cite[e.g.][]{chen_detecting_2014,baltagi_estimating_2021} who necessarily interpret and associate all breaks in the factor structure such as the Great Moderation as breaks in the factor loadings. Being able to associate the Great Moderation break with a change in the factor variance reconciles how one can understand the underlying dynamic factor model in \Cref{eqn:static_form2,eqn:factor_dynamics} with the broader literature. 
By disentangling changes in the factor variance and loadings, we are able to interpret the Great Moderation as a change in the variance of the factors, and thus more likely associated with a break in Equation \eqref{eqn:factor_dynamics}. 
This is consistent with interpretations advanced by \cite{sims_were_2006}, \cite{primiceri_time_2005}, and \cite{cogley_drifts_2005}, whose central claims are the Great Moderation was a result of smaller shocks hitting the macro-economy, and thus interpret the Great Moderation as a change in $\Sigma_\eta$ in Equation \eqref{eqn:factor_dynamics}.\footnote{Note that a break in Equation \eqref{eqn:factor_dynamics} can also imply a break in the $\Phi$'s, so our results are merely consistent, but do not confirm, that the break is in $\Sigma_\eta$, which is a much stronger claim.} 
By disentangling the break, our procedure appears successful in making this distinction, and we are thus able to reconcile the results from the factor modeling literature with the broader Great Moderation literature, wherein this documented evidence of breaks in the factor loadings may have conflated changes in the factor variance with the factor loadings.
In contrast, the break associated with the Great Recession appears more mixed; while we are able to find breaks in the factor loadings, evidence of breaks in the factor variance is more mixed, aligning with extant evidence in \cite{stock_disentangling_2012}. 

\subsection{Individual Test Results}
Besides the joint sup-$W$ test, we can also test for breaks in loadings individually. \Cref{tab:w_ind_test} presents the number of series where we can reject at least one of their factor loadings breaking, controlling for a possible break in the factor variance. 
For comparison, we also apply \cite{breitung_testing_2011}'s test for breaks in the factor loadings though we caution that direct comparison is not straightforward due to their use of \emph{pseudo} factors, and hence do not allow for changes in the factor variance. 
We nonetheless suggest two tentative conclusions. First, there is some, albeit weak, evidence that by accounting for changes in the factor variance, one may find (marginally) fewer breaks in the loadings associated with the Great Moderation. 
Second, when we consider the Great Recession break, there is some evidence that we are able to find more series experiencing a break in the factor loadings compared to the \cite{breitung_testing_2011} procedure. 
Nonetheless, it seems these results heavily rely on the estimated number of factors. 
In any case, putting aside issues of direct comparability of our procedure relative to that of \cite{breitung_testing_2011}, the issue of multiple-hypothesis testing is prevalent due to the high number of individual tests. 
Indeed, when we applied a Bonferroni-Holm adjustment, which is known to be conservative, to the individual tests, we unsurprisingly find far fewer breaks in the individual factor loadings, though the general pattern that we find relative to the \cite{breitung_testing_2011} procedure, as reported in \Cref{tab:w_ind_test}, remains. 
We therefore encourage users to use the joint sup-$W$-test, which is specifically designed to overcome the issue of multiple testing, especially if the joint sup-$W$-test is the relevant hypothesis of interest. 
If practitioners are interested in specifically testing for evidence of a break in a group of variables (e.g. prices), rather than conducting individual tests which may be afflicted by issues associated with multiple testing, one can instead opt to conduct a joint sup-$W$ test over a specific group of variables, where the summation can be taken over all series in that group when constructing the test statistic.

\begin{table}
\scriptsize
\caption{\label{tab:w_ind_test}Individual Series Loading Break Test Rejection Counts. }
\centering
\begin{threeparttable}
\begin{tabularx}{\textwidth}{c}
\hspace{0.2\textwidth}\begin{tabular}[t]{ccc}
\toprule
$\tilde{r}$ & Individual sup-$w_i$ & Breitung and Eickmeier (2011)\\
\midrule
\addlinespace[0.3em]
\multicolumn{3}{l}{\textbf{Great Moderation (1984 Q1), 1959 Q3 - 2008 Q3 Sample}}\\
\hspace{1em}1 & 1 & 7\\
\hspace{1em}2 & 13 & 13\\
\hspace{1em}3 & 19 & 29\\
\hspace{1em}4 & 23 & 29\\
\addlinespace[0.3em]
\multicolumn{3}{l}{\textbf{Great Recession (2008 Q3), 1984 Q2 - 2019 Q4 Sample}}\\
\hspace{1em}1 & 5 & 8\\
\hspace{1em}2 & 17 & 5\\
\hspace{1em}3 & 24 & 30\\
\hspace{1em}4 & 42 & 18\\
\bottomrule
\end{tabular}
\end{tabularx}
\begin{tablenotes}
\item \textit{Note: } 
\item Numbers in cells represent the count of rejections of the null hypothesis that the loadings of an individual series broke at the given break date, (treated as unknown, 5\% significance level). Total of 124 series in each subsample.
\end{tablenotes}
\end{threeparttable}
\end{table}

\subsection{Which variables experienced a break in their loadings?}
\label{sec:empirical:variance_decomposition}
To further understand the break associated with the Great Moderation and Great Recession, we explored which types of variables had breaks in loadings. 
In order to understand whether these breaks were important for understanding the variation in variables, we calculate an $R^2$ measure for each series subject to no restrictions, and with the restriction that there were no breaks in the loadings.\footnote{See Section B.2 in the Supplement for how we use \Cref{eqn:projection_ert:1} to do so.} Thus, these $R^2$ statistics should have a large (small) difference if the breaks in the loadings were (un)important.
\begin{figure}
\centering
\includegraphics[width=0.9\textwidth]{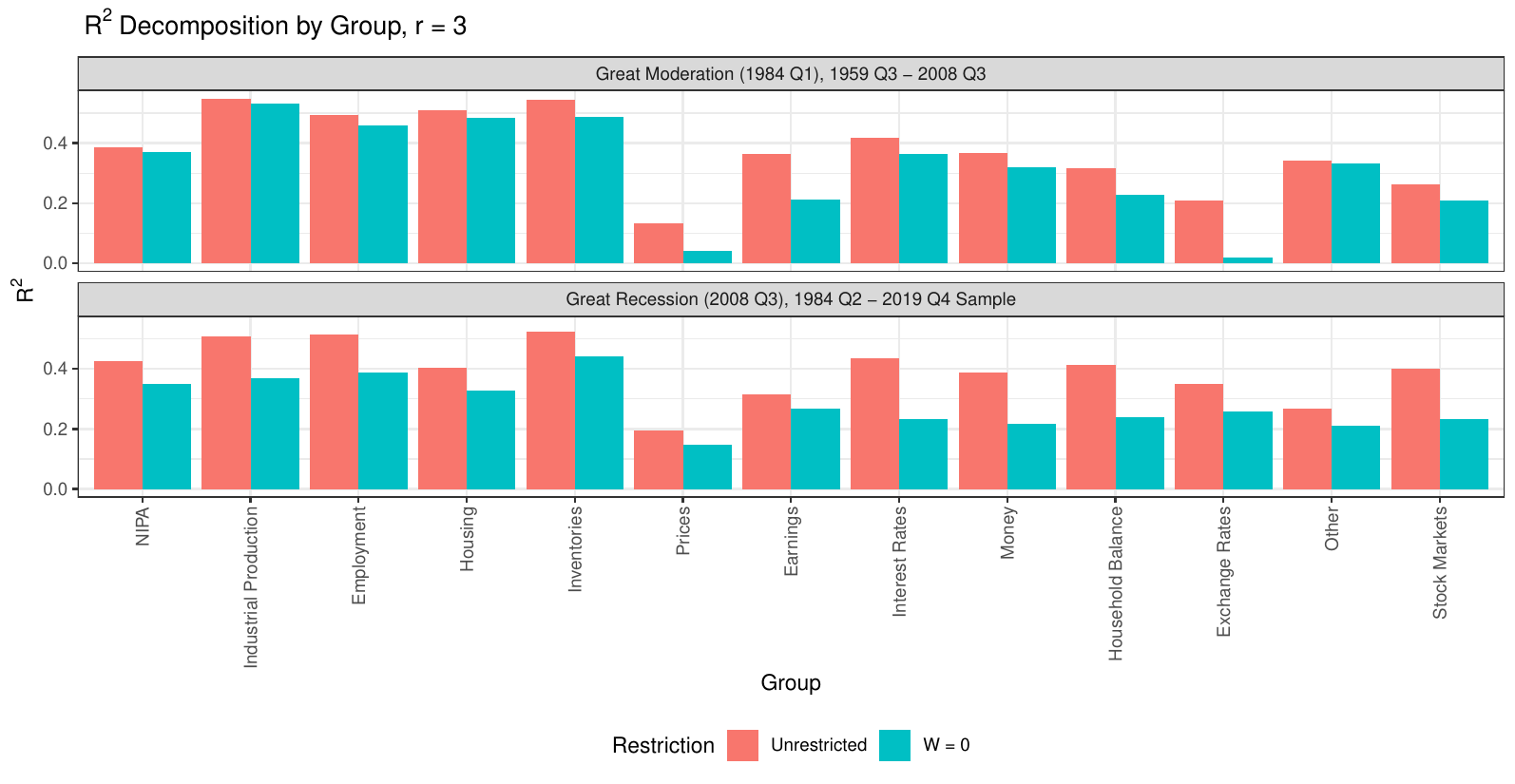}
\caption{$R^2$ Statistics for unrestricted and restricted common component ($W = \mathbf{0}$) for Great Moderation Subsample, and Great Recession Subsample, for $r = 3$. }
\label{fig:var_decomp_group}
\end{figure}
\Cref{fig:var_decomp_group} presents the unrestricted and $W = 0$ restricted $R^2$ statistics averaged across all series by category. We present the $R^2$ for $r = 3$ but note that these conclusions are very similar for $r = 1$ to $4$. For the Great Moderation break, breaks in the loadings appear to be more important for prices, earnings, exchange rates, and household balances. Two of these categories at least plausibly coincide with extant knowledge: the Great Inflation which preceded the Great Moderation and affected prices; and the collapse of the Bretton Woods system in the mid-1970s which affected exchange rate variables. For the Great Recession break, while it appears that breaks in the loadings were important for many variable categories, they appear important for financial variables in categories such as exchange rate, money, and the stock market. Additionally, it appears that breaks in loadings were important for prices, also documented by \cite{stock_disentangling_2012}.
\section{Conclusion}
The existing literature on structural breaks in factor models by and large does not distinguish between breaks between the factor variance and loadings, due to the need of a normalization during estimation. 
We argue it is important to distinguish them, as both can lead to different economic interpretations. 
To address this, we develop a projection-based decomposition of the structural break into a rotational and orthogonal shift component, which are naturally interpreted as a change in factor variance and loadings, respectively. 
The estimators are simple to calculate and lead to two test statistics which can be compared to standard supremum-type critical values in order to disentangle structural breaks in the factor variance and loadings. 
Their finite sample performance is confirmed by a Monte Carlo study. 
Applying our procedure to U.S. macroeconomic data, we show that the Great Moderation is more naturally associated with a break with the factor variance as opposed to the factor loadings. 
This is in contrast to extant methods which do not distinguish the two and would, by construction, associate these breaks with the factor loadings. 
Our projection-based decomposition allows us to estimate that the Great Moderation is associated with an over 70\% reduction in the factor variance, a result precluded \emph{a priori} if the break is not disentangled, and thus highlights the importance of doing so. 
Our results thus reconciles factor model work which find breaks in the factor loadings with the broader literature on the Great Moderation, which more naturally interprets the Great Moderation as a change in the common variance across variables.

Our framework, in principle, permits the decomposition of any break in the factor structure into variance and loading components, which may be helpful in modeling and/or parameterizing breaks in factor models, including those which are estimated parametrically using state-space methods. 
Furthermore, natural questions remain on how different break types can affect the subsequent use of factors in common applications, such as  factor-augmented forecasts or vector auto-regressions. 
These represent interesting avenues for future work to build on the ideas that we have developed.

%


\bibliographystyle{apalike}
\small
\bibliography{bibtex.bib}

\end{document}